\documentclass{article}
\usepackage{amsmath,amsthm,amssymb,graphicx,subfigure}
\usepackage[hmargin=2cm,vmargin=3.5cm]{geometry}

\newtheoremstyle{asdf}% name
  {}%      Space above
  {}%      Space below
  {}%         Body font
  {}%         Indent amount (empty = no indent, \parindent = para indent)
  {\bfseries}% Thm head font
  {:}%        Punctuation after thm head
  {\newline}%     Space after thm head: " " = normal interword space;
  {}%         Thm head spec (can be left empty, meaning `normal')

\newtheorem{thm}{Theorem}[section]

\newtheorem{prop}[thm]{Proposition}

\newtheorem{defn}{Definition}
\newtheorem{remark}{Remark}

\theoremstyle{asdf}

\let\hat\widehat

\title{Random Differential Privacy}
\author{Rob Hall \and Alessandro Rinaldo \and Larry Wasserman}

\begin{document}

\maketitle

\begin{abstract}
We propose a relaxed privacy definition
called {\em random differential privacy} (RDP).
Differential privacy requires that adding
any new observation to a database will have small
effect on the output of the data-release procedure.
Random differential privacy requires that adding
a {\em randomly drawn new observation} to a database will have small
effect on the output.
We show an analog of the
composition property of differentially private procedures which
applies to our new definition.
We show how to release an RDP histogram
and we show that RDP histograms are much more accurate
than histograms obtained
using ordinary differential privacy.
We finally show an analog of the global sensitivity framework for the release of functions under our privacy definition.
\end{abstract}

\section{Introduction}

{\em Differential privacy} (DP)
(\cite{Dwork:06})
is a type of privacy guarantee that has
become quite popular in the computer science literature.  The
advantage of differential privacy is that it gives a strong and
mathematically rigorous guarantee.  The disadvantage is that the
strong privacy guarantee often comes at the expense of the statistical utility of the released information.
We propose a weaker
notion of privacy,
called ``random differential privacy'' (RDP),
under which it is possible to achieve better accuracy.

The privacy guarantee provided by RDP represents a radical
weakening of the ordinary differential privacy.
This could be a cause for concern for those who want very
strong privacy guarantees.
Indeed, {\em we are not suggesting the RDP should replace
ordinary differential privacy.}
However, as we shall show in this paper
(and has been observed many times in the past),
differential privacy can lead to large information losses
in some cases (see e.g., \cite{xiaolin}).
Thus, we feel there is great value
in exploring weakened versions of differential privacy.
In other words,
we are proposing a new privacy definition as a way of
exploring the privacy/accuracy tradeoff.

We begin by introducing ordinary differential privacy and setting up some notation.  We then explore the lower limits for accuracy of differentially private techniques in the context of histograms.  We introduce a concept which parallels minimaxity in statistics, and identify the minimax risk for a differentially private histogram.  We describe an important subset of these minimax differentially private histograms which we show to have risk which is uniformly lower bounded at a rate which is linear in the dimension of the histogram.  We then introduce our proposed relaxation to differential privacy, under which our technique enjoys the same minimax risk, but with a lower bound which depends only on the size of the support of the histogram (namely, the number of nonzero cells).  Thus we show that in the context of sparse histograms, the relaxation allows for a strictly better data release.  We also demonstrate some important properties of our relaxation, such as an analog of the composition lemma.

\section{Differential Privacy (DP)}

\subsection{Definition}

Let $X=(X_1,\ldots, X_n)\in {\cal X}^n$ be an input database with $n$ observations
where $X_i \in {\cal X}$.
The goal is to produce some output $Z\in {\cal Z}$.  For example the inputs may consist of database rows in which each column is a measurement of an individual, and the output is the number of individuals having some property.
Let $Q_n(\ \cdot\ |X)$ be a conditional distribution for $Z$ given $X$.
Write $X\sim X'$ if
$X,X'\in {\cal X}^n$ and
$X$ and $X'$ differ in one coordinate.
We say that $X$ and $X'$ are
{\em neighboring databases.}
\footnote{In some papers, the definition is changed
so that one sample is a strict subset of the other, having exactly one
less element.  Although this definition is perhaps slightly stronger,
we do not use it and remark that the approaches we present below may
all be fit into this framework if so desired.}

We say $Q_n$ satisfies $\alpha$ differential privacy if,
for all measurable $B\subset {\cal Z}$
and all $X\sim X' \in {\cal X}^n$,
\begin{equation}
e^{-\alpha} \leq
\frac{Q_n(Z\in B|X)}{Q_n(Z\in B|X')} \leq e^\alpha.
\end{equation}
The intuition is that, for small $\alpha>0$,
the value of one individual's data has small effect
on the output.
We consider any DP
algorithm to be a family of distributions
$\mathbb{Q}_n$ over the output space $\mathcal{Z}$.
We index a
family of distributions by $n$ to show the size of the dataset.

It has been shown by researchers in privacy that
differential privacy provides a very strong guarantee.
Essentially it means that whether or not one particular individual
is entered in the database,
has negligible effect on the output.
The research in differential privacy is vast.
A few key references are
\cite{Dwork:06},
\cite{Dwork:07},
\cite{barak2007paa},
\cite{DL09},
\cite{BDMN05} and references therein.

\subsection{Noninteractive Privacy and Histograms}

Much research on differential privacy
focuses on the case where
$Z$ is a response to some query such as
``what is the mean of the data.''
A simple way to achieve differential privacy in that case
is to add some noise to the mean of $X$
where the noise has a Laplace distribution.
The user may send a sequence of such queries.
This is called {\em interactive privacy}.
We instead focus on the
{\em noninteractive privacy}
where the goal is to output a whole database (or a ``synthetic dataset'')
$Z = (Z_1,\ldots, Z_N)$.
Then the user is not restricted to a small number of queries.

One way to release a private database
is to first release a privatized
histogram.
We can then draw an arbitrarily large sample
$Z=(Z_1,\ldots, Z_N)$ from the histogram.
It is easy to show that if the histogram
satisfies DP then $Z$ also satisfies DP.
Hence, in the rest of the paper,
we focus on constructing
a private histogram.

We consider privatization mechanisms which are permutation invariant with respect to their inputs (i.e.,
those distributions which treat the values $x_i$ as a set rather than a
vector)  in the context of histograms this appears to be a very mild restriction.  

We partition the sample space ${\cal X}$ into $k$ cells (or bins) $\{B_j\}_{j=1}^k$.\footnote{
In this paper, $k$ is taken as a given integer.
The problem of choosing an optimal $k$ in a private matter
is the subject of future work.}
We consider the input to be a lattice point in the
$k$-simplex, by taking the function:
$\theta^n(x_1,\ldots,x_n) = (\theta_1,\ldots,\theta_k)$,
$\theta_j = \frac{1}{n}\sum_{i=1}^n\mathbf{1}\{x_i \in B_j\}.$
The image of this mapping $\Theta = \theta^n(\mathcal{X}^n)$ is the set of
lattice points in the simplex which correspond to histograms of $n$
observations in $k$ bins. Note that this is in essence a ``normalized histogram'' since the elements sum to one.  This set depends on $k$ although we suppress this notation.  For the remainder of this paper we consider the output space $\mathcal{Z}$ to be the same as the input space (i.e., a normalized histogram).  
%For $x$ in bin $B_j$ define the histogram density estimator
%$\hat f(x) = \hat \theta_j /{\rm Volume}(B_j)$.

Now we give a concrete example of a $Q_n$ which achieves differential privacy.  Define
$z_j = \theta_j + 2L_j/(n\alpha)$
where $L_1,\ldots, L_k$ are independent draws
from a Laplace distribution with mean zero and rate one.
Then $(z_1,\ldots, z_k)$ satisfy DP (see e.g.,\cite{Dwork:06}).
However, the $z_i$ themselves do not represent a histogram,
because they can be negative and they do not necessarily sum to one.
Hence we may take, for example:

\begin{equation}
\label{delta_z}
\delta(z) = \arg\min_{\theta \in \Theta}\|z-\theta\|_1
\end{equation}

\noindent where we use the $\ell_1$ norm:
$||x||_1 = \sum_j |x_j|$. This procedure hence results in a valid histogram.  Note that $\delta(z)$ satisfies the differential privacy, since each subset of values it may take clearly corresponds to a measurable subset of $\mathbb{R}^k$.  Since the differential privacy held for the real vector then it also holds for the projection (see e.g., \cite{WZ}).  We will refer to this as the
{\em histogram perturbation method}
(see e.g., \cite{WZ}).
There are other methods for generating
differentially private histograms, and our results below concern hold over a large subset of all the possible techniques available (to be made precise after proposition 3.2).  Hence our results apply to more than the above concrete scheme.

\section{Lower Bounds for Accuracy with Differential Privacy}

To motivate the need for relaxed versions of
differential privacy, we consider here the accuracy
of differentially private histograms.  We evaluate a differentially private procedure in terms of its ``risk'' which is a natural measure of accuracy taken from statistics.  We consider the $\ell_1$ loss function, and the associated risk:

\begin{equation}
\label{risk}
R(\theta,Q_n) = \int_{\Theta} \|\hat{\theta} - \theta\|_1 dQ_n(\hat{\theta}|\theta).
\end{equation}

\noindent where $\hat{\theta}$ is the output of the differentially private algorithm, $\theta$ is the input histogram, and the distribution $Q_n$ is the one induced by the randomized algorithm.  Typically this risk will be a non-constant function of the parameter $\theta$ and of the distribution $Q_n$.  Therefore we consider the ``minimax risk'' which is the smallest achievable worst-case risk, and gives a measure of the hardness of the problem which does not depend on a particular choice of procedure:

\begin{equation}
\label{minimax_risk}
R^\star = \inf_{Q_n}\sup_{\theta \in \Theta} R(\theta,Q_n)
\end{equation}

We next describe the minimax risk of
the best fully differentially private mechanism $Q_n$.  

\begin{prop}
$$
R^\star \geq c_0\frac{k-1}{\alpha n}$$
\end{prop}
\begin{proof}
The proof uses a standard method for deriving minimax lower bounds in statistical estimation. Consider the $k-1$- dimensional hypercube
$$
\left\{ \left( \frac{\sigma_1\tau}{n},\ldots,
\frac{\sigma_{k-1}\tau}{n}, \frac{(n -\sum_{i=1}^{k-1}\sigma_i)\tau}{n} \right) :
\sigma_i \in \{0,1\}
\right\}.
$$

Take $\theta,\theta^\prime$, to be neighboring corners of
this hypercube (namely two elements which differ in exactly one coordinate
$\sigma_i$).  Take the KL divergence between the conditional distributions at these corners to be:

$$KL\left(Q_n(\cdot|\theta)\big\|Q_n(\cdot|\theta^\prime)\right) = \int_{\Theta} \log\frac{Q_n(\hat{\theta}|\theta)}{Q_n(\hat{\theta}|\theta^\prime)}dQ_n(\hat{\theta}|\theta)$$

\noindent By considering a sequence of points corresponding to neighboring inputs, we find the ratio of densities to have the upper bound: $\frac{Q_n(\hat{\theta}|\theta)}{Q_n(\hat{\theta}|\theta^\prime)} \leq e^{\alpha\tau}$ since $\tau$ elements of the input have to change to move from $\theta$ to $\theta^\prime$, and the ratio at each step is bounded by $e^\alpha$.  Therefore the KL divergence obeys $KL\left(Q_n(\cdot|\theta)\big\|Q_n(\cdot|\theta^\prime)\right) \leq \alpha\tau.$ The ``affinity'' between the two distributions is:

$$
\|Q_n(\cdot|\theta) \wedge Q_n(\cdot|\theta^\prime)\| = \int_\Theta \min\left\{Q_n(\hat{\theta}|\theta),Q_n(\hat{\theta}|\theta^\prime)\right\} d\hat{\theta}.
$$

The Kullback-Csiszar-Kemperman inequality \cite{bin_yu} yields
a lower bound on the affinity between these distributions:
$$
\|Q_n(\cdot|\theta) \wedge Q_n(\cdot|\theta^\prime)\| \geq 1 - \sqrt{\frac{\alpha\tau}{2}}.
$$
Assouad's lemma (see \cite{bin_yu} again) thus gives the lower bound:
$$
R^\star \geq (k-1)\frac{\tau}{2n}\left( 1 - \sqrt{\frac{\alpha\tau}{2}}  \right).
$$
Taking $\tau = t/\alpha$ gives
$$
R^\star \geq (k-1)\frac{t}{2\alpha n}\left( 1 - \sqrt{\frac{t}{2}} \right).
$$
For $\alpha < 1$ we may take $t < 1$, which results in the parenthetical expression being positive.
\end{proof}

\begin{remark}
The previous result demonstrates that the minimax risk of the
differentially private histogram is of the order
$O\left(\frac{k}{\alpha n}\right)$.
\end{remark}

\begin{remark}
Hardt and Talwar \cite{HT} have a similar result although
their setting is somewhat different.
In particular, they do not restrict to
the space of histograms based on $n$ observations.
\end{remark}

The above results demonstrates that for every differentially private scheme, there is at least one input for which the risk is growing in the order shown (in fact, at least one point in every hypercube of side length $\tau/n$).  However the prospect exists that at many other inputs the risk is much lower.  We now demonstrate that this is not the case when $k=2$, by presenting a uniform lower bound for the risk among all minimax schemes.  In the case of $k=2$ the output may be regarded as a single number $\frac{a}{n}$ where $a \in \{0,\ldots,n\}$,
which gives the proportion of the data points in the first bin.  Our result will show that in a sense, the minimax differential privacy schemes are similar to ``equalizer rules'' in the sense that the risk is on the same order for every input.

\begin{prop}
For $k=2$ for any $Q_n$ which achieves $\sup_\theta R(\theta,Q_n) \leq \frac{c_0}{\alpha n}$ we have that $\inf_\theta R(\theta,Q_n) \geq \frac{c_1}{\alpha n}$
\end{prop}
\begin{proof}
Note that for any $\theta_1$ and $c>c_0$, due to the uniform upper bound on the risk, Markov's inequality gives

$$\int_\mathcal{Z} \mathbf{1}\{|\hat{\theta}-\theta_1|\leq \frac{c}{\alpha n}\} \ dQ_n(\hat{\theta}|\theta_1) \geq 1 - \frac{c_0}{c}.$$
Therefore, due to the constraint of differential privacy, we have that, for any $\theta_0$,

$$\int_\mathcal{Z} \mathbf{1}\{|\hat{\theta}-\theta_1|\leq \frac{c}{\alpha n}\}\ dQ_n(\hat{\theta}|\theta_0) \geq \left(1 - \frac{c_0}{c}\right)\text{exp}\left\{-\frac{\alpha n}{2}\|\theta_0-\theta_1\|_1 \right\}$$

\noindent Since $\frac{n}{2}\|\theta_0-\theta_1\|$ elements of the input change to move from $\theta_0$ to $\theta_1$.  Therefore taking $\theta_1$ to give $\|\theta_0-\theta_1\| = \frac{2c}{\alpha n}$ gives

$$R(\theta_0,Q_n) \geq \frac{c}{\alpha n} \left(1 - \frac{c_0}{c}\right)e^{-c} = \frac{c_1}{\alpha n}.$$

As $\theta_0$ is arbitrary, this gives a uniform lower bound under the conditions above.
\end{proof}

For the relaxation of differential privacy given in definition 2.2 of
\cite{HT}, the above result remains intact for large enough $n$.  The
relaxation is:
$$
Q_n(z|X) \leq Q_n(z|X')e^{\alpha} + \eta(n)
$$
\noindent where $\eta(n)$ is negligible (i.e., tending to zero faster than
any inverse polynomial in $n$).  Thus via the same technique as above,
we have
$$
R(\theta_0,\delta,\mathbb{Q}_n) \geq
\frac{c}{\alpha n} \left((1 - \frac{c_0}{c})e^{-c} - c_2\eta(n)\right) =
\frac{c_1 - \eta(n)}{\alpha n}.
$$
For large enough $n$ this latter term is
bounded from below by $\frac{c_3}{\alpha n}$.  This indicates that the
above relaxation of differential privacy will not be
useful in achieving higher accuracy.

For $k > 2$, we may write

$$
R(\theta,Q_n) =
\sum_{i=1}^kR_i(\theta,Q_n)
$$

\noindent With

$$
R_i(\theta,Q_n) \stackrel{\text{def}}{=}
\int_{\mathcal{Z}} |\hat{\theta} - \theta_i| dQ_n(\hat{\theta}|\theta),
$$

where the subscript means the $i^{th}$ coordinate.  Thus, whenever we have
that $R_i \leq \frac{c_0}{\alpha n}$ uniformly over $i$, we have that
$R(\theta,\delta,\mathbb{Q}_n) \geq \frac{c_1(k-1)}{\alpha n}$.
Therefore the only opportunity to improve upon the rate of
$\frac{k}{\alpha n}$ is when some $\theta$ have some coordinate $i$ at
which the risk upper bound does not apply.

We conclude by remarking that we have demonstrated, that for a certain class of differentially private algorithms which achieve the ``minimax rate,'' their risk is uniformly lower bounded at the same rate.  The rate in question is linear in $k$, which is problematic when $k$ is large relative to $n$.  It remains an open question whether there are different techniques which achieve the minimax rate, yet do not have this property.  Such a technique would have to lose the uniform upper bound on the coordinate-wise risk.  Below, we present a weakening of differential privacy, which admits release mechanisms, which both keep the uniform upper bound on the coordinate-wise risk, and also have a minimax risk which is growing only in the support of the histogram (namely, the number of cells which contain observations).

%A corollary of the above
%is the following explicit bound
%in the super-sparse case
%$\hat\theta = (1,0,\ldots, 0)$.
%Although the result can be deduced from the above,
%we provide a direct proof sketch in the appendix.
%
%\begin{thm}
%\label{thm::extra}
%Suppose that $n\to \infty$,
%$k = k(n) \to \infty$ and $k/n \to 0$.
%Let $\hat\theta = (1,0,\ldots, 0)$.
%Then
%$$
%\mathbb{E}_Q ||\hat\theta - \hat\theta^*||_1 \geq
%\frac{k}{n\alpha} + o\left(\frac{k}{n\alpha}\right).
%$$
%\end{thm}

%We shall now introduce random differential privacy and show it allows
%for higher accuracy even in sparse cases, and without sacrificing the
%coordinate-wise uniform upper bound on the risk.  It remains an open
%question of whether it is possible to achieve a lower rate for some
%histograms, within the confines of differential privacy, with a
%mechanism which still achieves the minimax rate presented above.  The
%above concerns only exclude an important subset of such mechanisms.

\section{Random Differential Privacy}

In random differential privacy (RDP) we view the
data $X=(X_1,\ldots, X_n)$
as random draws from an unknown distribution $P$.
This is certainly the case in statistical sampling and of course
it is the usual assumption in most learning theory.
Let us denote the observed values of the random variables
$X=(X_1,\ldots, X_n)$ by
$x=(x_1,\ldots,x_n)$.
Recall that under DP,
$Q(Z\in B|x_1,\ldots,x_n)$
is not strongly affected if we replace some value
$x_i$ with another value $x_i'$.
We continue to restrict to the case in which,
$Q(Z\in B|x_1,\ldots,x_n)$
is invariant to permutations of
$(x_1,\ldots,x_n)$.
Thus we may restate DP by saying that
$Q(Z\in B|x_1,\ldots,x_n)$
is not strongly affected if we replace $x_n$ by some other
arbitrary value $x_n'$.
In RDP, we require instead that the distribution $Q_n(\cdot|x_1,\ldots,x_n)$
is not strongly affected if we replace $x_n$ by some
new $x_n'$ which is also randomly drawn from $P$.

\begin{defn}[$(\alpha,\gamma)$-Random Differential Privacy]
We say that a randomized algorithm $Q_n$ is $(\alpha,\gamma)$-Randomly
Differentially Private when:
$$
\mathbb{P}\left( \forall B \subseteq \mathcal{Z}, \
e^{-\alpha} \leq
\frac{Q_{n}(Z\in B|X)}
{Q_n(Z\in B|X^\prime)} \leq e^\alpha \right) \geq 1-\gamma
$$
\noindent where 
$$X = (X_1,\ldots,X_{n-1},X_n),\ X^\prime = (X_1,\ldots,X_{n-1},X_{n+1})$$
\noindent (i.e., $X \sim X^\prime$), and 
the probability is with respect to the $n+1$-fold product measure
$P^{n+1}$ on the space $\mathcal{X}^{n+1}$, that is,
$X_1,\ldots, X_{n+1} \stackrel{\text{iid}}{\sim} P$.
\end{defn}

\noindent We also give the ``random'' analog of the $(\alpha,\delta)$-Differential Privacy:

\begin{defn}[$(\alpha,\eta,\gamma)$-Random Differential Privacy]
We say that a randomized algorithm $Q_n$ is $(\alpha,\eta,\gamma)$-Randomly
Differentially Private when:
$$
\mathbb{P}\left( \forall B \subseteq \mathcal{Z}, \
Q_{n}(Z\in B|X) \leq e^\alpha Q_n(Z\in B|X^\prime) + \eta(n) \right) \geq 1-\gamma
$$
\noindent where $\eta$ is negligible (i.e., decreasing faster than any inverse polynomial).
\end{defn}

We note that
\cite{MKAGV08} also
consider a probabilistic relaxation of
DP. However, their relaxation is quite different than
the one considered here.  Namely, their relaxation bounds the probability that the differential privacy criteria is not met, but where the probability is taken with respect to the randomized algorithm itself.  Our relaxation takes the probability with respect to the generation of the data itself.  
The following result is clear from the definition of random differential privacy.

\begin{prop}
$(\alpha,\gamma)$-RDP is a strict relaxation of $\alpha$-DP.
That is, if $Q_n$ is DP then it is also RDP.
However, there are RDP procedures that are not DP.
\end{prop}

\begin{remark}
Although an $\alpha$-DP procedure fulfils the requirement of
$(\alpha,0)$-RDP, the converse is not true.  The reason is that the
latter requires that the condition (that the ratio of densities be
bounded) holds almost everywhere with respect to the unknown
measure, whereas DP require that this condition holds uniformly
everywhere in the space.
\end{remark}

We next show an important property of the definition, namely, that
RDP algorithms may be composed to give other RDP algorithms with
different constants.  The analogous composition property for DP was
considered to be important because it allowed rapid development of
techniques which release multiple statistics, as well as techniques
which allow interactive access to the data.

\begin{prop}[Composition]
Suppose $Q,Q^\prime$ are distributions over
$\mathcal{Z},\mathcal{Z}^\prime$ which are $(\alpha,\gamma)$-RDP and
$(\alpha^\prime,\gamma^\prime)$-RDP respectively.  The following
distribution $C$ over $\mathcal{Z}\times \mathcal{Z}^\prime$ is
$(\alpha+\alpha^\prime,\gamma+\gamma^\prime)$-RDP:
$$
C(Z,Z^\prime|X) = Q(Z|X)\cdot Q^\prime(Z^\prime|X).
$$
\end{prop}

This result is simply an application of the union bound combined with
the standard composition property of differential privacy.  As an
example, suppose it is required to release $k$ different statistics of
some data sample.  If each one is released via a
$(\alpha/k,\gamma/k)$-RDP procedure, then the overall release of all
$k$ statistics together achieves $(\alpha,\gamma)$-RDP.  A similar result holds for the composition of $(\alpha,\delta,\gamma)$-RDP releases.

\section{RDP Sparse Histograms}

We first give a technique for the release of a histogram which works
well in the case of a sparse histogram, and which satisfies the
$(\alpha,\gamma)$-Random Differential Privacy.  We then compare the
accuracy of this method to a lower bound on the accuracy of a
$\alpha$-Differentially Private approach.

The basic idea is to not add any noise to cells with low counts.
This results in partitioning the space into two blocks and releasing a
noise-free histogram in one block, and use a differentially private
histogram in the other.  The partition will depend on the data itself.
For a sample $x_1,\ldots,x_n$, we denote:
$S = S(x_1,\ldots,x_n) = \left\{j : \theta_j = 0 \right\}.$

\noindent Then we consider the release mechanism:
\begin{equation}\label{eqn_hist_counting}z_j = \begin{cases}
\theta_j & j \in S \text { and } 2k \leq \gamma n \\
\theta_j + \frac{2}{n\alpha}L & \text{o/w}
\end{cases}\end{equation}

\begin{prop}
The random vector
$Z = (z_1,\ldots,z_k)$
as defined in (\ref{eqn_hist_counting}) satisfies the $(\alpha,\gamma)$-RDP.
\end{prop}

\noindent In demonstrating RDP, we take the sample
$x_1,\ldots,x_n,x_{n+1}$ and denote: $S = S(x_1,\ldots,x_n)$ and
$S^\prime = S(x_1,\ldots,x_{n-1},x_{n+1})$.  We consider the output
distribution of our method when applied to each of the neighboring
samples.  The event that the ratio of densities fail to meet the
requisite bound is a subset of the event where either $x_{n+1} \in S$
or $x_n\in S^\prime$, and when $2k \leq \gamma n$.  In the complement
of this event then the partitions are the same, and the differing
samples both fall within the block which receives the Laplace noise,
so the DP condition is achieved.  In demonstrating the RDP, we simply
bound the probability of the aforementioned event, conditional on the
order statistics.

\begin{proof}[Proof of proposition 5.1]
In the interest of space let the vector of order statistics be denoted $T = (x_{(1)},\ldots,x_{(n+1)})$.
Let
$S^\star(x_1,\ldots,x_n,x_{n+1}) = \left\{j : \sum_{i=1}^{n+1}\mathbf{1}\{x_i=j\} \leq 1 \right\}$.
We have that $S,S^\prime \subseteq S^\star$. We thus have
$$
\mathbb{P}(x_n\in S^\prime \text{ or } x_{n+1} \in S | T )
\leq 
\mathbb{P}(x_n\in S^\star \text{ or } x_{n+1} \in S^\star| T ).
$$
The latter probability is just the fraction of ways in which the order
statistics may be rearranged so that $x_n,x_{n+1}$ fall within
$S^\star$.  Due to the condition $2k \leq \gamma n$, we have
$|S^\star| \leq k \leq \frac{\gamma n}{2}$.  Therefore the number of
rearrangements having at least one of $x_n$ or $x_{n+1}$ in $S^\star$ is
bounded above
$$
\mathbb{P}\left(x_n\in S^\star \text{ or } x_{n+1} \in S^\star| T \right)
\leq \frac{2|S^\star|}{n+1} < \gamma.
$$
Therefore
\begin{align*}
\mathbb{P}(x_n\in S^\prime \text{ or } x_{n+1} \in S ) &\leq  \int_{\mathcal{X}^{n+1}}\mathbb{P}(x_n\in S^\prime \text{ or } x_{n+1} \in S | T )dP(T)
&\ leq \int_{\mathcal{X}^{n+1}}\mathbb{P}(x_n\in S^\star \text{ or } x_{n+1} \in S^\star | T )dP(T) \\
& < \gamma \int_{\mathcal{X}^{n+1}}dP(T)  \\
& = \gamma.
\end{align*}
Finally:
\begin{align*}
\mathbb{P}\left( \forall Z \subseteq \mathcal{Z}, e^{-\alpha} \leq \frac{Q_{n}(Z|X)}{Q_n(Z|X^\prime)} \leq e^\alpha \right)
&= 1 - \mathbb{P}(x_n\in S^\prime \text{ or } x_{n+1} \in S ) \\
&> 1 - \gamma.
\end{align*}
\end{proof}

\subsection{Accuracy}

Here we show that
$\delta(z)$ from (\ref{delta_z}) is close to $\theta$
even when the histogram is sparse.

%First consider the super-sparse case
%$\hat\theta = (1,0,\ldots, 0)$.
%We saw earlier that for ordinary
%DP we have
%$\mathbb{E} ||\hat \theta - \hat\theta^*||_1 \geq k/(n\alpha)$.
%But for RDP, it is clear that, for large $n$,
%$||\hat \theta - \tilde\theta||_1 =0$
%showing a dramatic increase in accuracy.

\begin{thm}
Suppose that $2k \leq \gamma n$.
Let $\theta^n(x_1,\ldots,x_n) = (\theta_1,\ldots, \theta_r, 0,\ldots, 0)$
for some $1\leq r < k$.
Then
$||\theta - \delta(z)||_1 = O_P(r/\alpha n)$.
\end{thm}

\begin{proof}
Let $L_1,\ldots, L_r \sim {\rm Laplace}$.
Let ${\cal E}$ be the event that
$L_j > - \frac{n \alpha}{2} \theta_j$ for all $1\leq j \leq r$.
Then ${\cal E}$ holds, except on a set of exponentially small probability.
Suppose ${\cal E}$ holds.
Let
$W = \sum_{j=1}^r L_j = O_P(r)$.
For $1\leq j \leq r$,
$z_j = \Bigl(\theta_j + (2L_j)/(n\alpha)\Bigr)$
For $j>r$,
$z_j = \theta_j =0$.
Hence $\|z-\theta\|_1 = O_P(r/\alpha n)$.
Furthermore $\|\delta(z)-z\|_1 \leq \frac{r}{n} \leq \frac{r}{\alpha n}$
Hence via the triangle inequality we have,
$||\delta(z) - \theta||_1 = O_P(r/\alpha n)$.
\end{proof}

We thus have a technique for which the risk is uniformly bounded above by $O(k/\alpha n)$ as with the DP technique, and which also enjoys the coordinate-wise upper bound on the risk.  However in this regime, the risk is no longer uniformly lower bounded with a rate linear in $k$, since the upper bound is linear in $r$ in the case of sparse vectors.

%%%IMO you can just say that on the nonzeros you just add laplace noise,
%%%so you get k'/an where k' = number of nonzeros.  Then you can say that
%%%the projection back to a lattice point in the simplex adds no more
%%%than k'/n (i.e., when the closest simplex point is in the opposite
%%%direction from the true thing), and that this last term is < k'/an for
%%%a<1, so the overall rate is c*k'/an. - R
%%%
%%%Assume that $2 k \leq \gamma n$ and let $s = |S|$. Then, for any $j \in S$,
%%%\[
%%%\tilde{\theta}_j = \frac{\left[ \hat{\theta}_j + L_j \frac{2}{\alpha n} \right]_{+}}{D},
%%%\]
%%%where $\{ L_j\}_{j \in S}$ are independent Laplace random variables with unit scale and
%%%\[
%%%D = \sum_{j \in S}\left[\hat{\theta}_j + L_j \frac{2}{\alpha n} \right]_{+}.
%%%\]
%%%Letting $\hat{\theta}_{\min} = \min_{j \in S} \hat{\theta}_j$, direct evaluation of the tail probability of a Laplace distribution and the union bound yield that $\hat{\theta}_j + L_j \frac{2}{\alpha n} > 0$ for all $j\in S$ with probability at least
%%%\[
%%%1 - \frac{s}{2} e^{- \hat{\theta}_{\min} n \alpha /2 }.
%%%\]
%%%Need to show that $D$ concentrates around its expectation. This is trickier than before (the value of the expectation is also trickier)
%%%
%%% Then,
%%%\[
%%%\sum_{j} |\tilde{\theta}_j - \hat{\theta}_j| = \sum_{j \in S} |\tilde{\theta}_j - \hat{\theta}_j| \approx
%%%\]
%%%
%%%\end{proof}
%%%

\section{RDP via Sensitivity Analysis}

We next demonstrate that RDP allows schemes for release of other kinds of statistics (besides histograms).  A common technique used to establish a differentially private technique is to use Laplace noise with variance proportional to the ``global sensitivity'' of the function \cite{DMNS06}.  We show that there is an analog of this technique for RDP.  We next demonstrate a method for the RDP release of an arbitrary function $g_n(x_1,\ldots,x_n) \in \mathbb{R}$.

We consider the algorithm which samples the distribution

\begin{equation}
\label{rdp_laplace}
Q_n(z|x_1,\ldots,x_n) \propto \text{exp}\left\{ \frac{-\alpha\left|z-g_n(x_1,\ldots,x_n)\right|}{s_n(x_1,\ldots,x_n)}\right\}
\end{equation}

It is well known that when $s_n$ is the constant function which gives an upper bound of the global sensitivity \cite{DMNS06} of $g_n$, this method enjoys the $\alpha$-DP.  As we allow $s_n$ to depend on the data we may make use of the local sensitivity framework of \cite{NRS07}.  There it is demonstrated that whenever:

\begin{equation}
\label{local_cond1}
\forall X \sim X^\prime\ s_n(X) \leq e^\beta s_n(X^\prime)
\end{equation}

\noindent and

\begin{equation}
\label{local_cond2}
\forall X\ \sup_{X^\prime \sim X}\left|g_n(X) - g_n(X^\prime)\right| \leq s_n(X)
\end{equation}

\noindent then (\ref{rdp_laplace}) gives $(2\alpha,\eta)$-DP with:

\begin{equation}
\label{local_eta}
\eta = e^{-\frac{\alpha}{2\beta}}
\end{equation}

\noindent (see \cite{NRS07} definition 2.1, lemma 2.5 and example 3).  In moving from DP to RDP we may now require that conditions (\ref{local_cond1}) and (\ref{local_cond2}) hold only with the requisite probability $1-\gamma$.  Then (\ref{rdp_laplace}) will achieve $(2\alpha,\eta,\gamma)$-RDP.

We consider a special subset of functions for which:

$$\sup_{X \sim X^\prime}\left|g_n(X) - g_n(X^\prime)\right| = n^{-1}\sup_{x,x^\prime}h(x,x^\prime).$$

\noindent Examples of functions satisfying this property are e.g., statistical point estimators \cite{smith:2008} and regularized logistic regression estimates \cite{pplr}. In particular in these cases it is assumed that $\mathcal{X}$ is some compact subset of $\mathbb{R}^d$ and then e.g., $\sup_{x,x^\prime}h(x,x^\prime) = \|x - x^\prime\|_2$ gives the diameter of this set.  

We replace conditions (\ref{local_cond1}) and (\ref{local_cond2}) with:

\begin{equation}
\label{rdp_cond1}
P\left(s_n(X) \leq e^\beta s_n(X^\prime)\right) \geq 1-\gamma_1
\end{equation}

\noindent and

\begin{equation}
\label{rdp_cond2}
P\left(n^{-1}h(x,x^\prime) \leq \min\{s_n(X),s_n(X^\prime)\} \right) \geq 1-\gamma_2.
\end{equation}

\noindent Note that $x,x^\prime$ are random draws from $P$ which are independent of the random vectors $X,X^\prime$.  The first condition simply requires (\ref{local_cond1}), to hold except on a set of measure $\gamma_1$.  The second condition implies that both $s_n(X)$ and $s_n(X^\prime)$ give upper bounds to the local sensitivity, except on a set of measure $\gamma_2$.  Putting these together along with the above considerations will yield a $(2\alpha,\eta,\gamma_1+\gamma_2)$-RDP method.  We note that we are essentially asking that $s_n(X)$ and $s_n(X^\prime)$ both give valid quantiles for the random variable $h(x,x^\prime)$, and that they give similar values with high probability.

We consider the empirical process based on $h$ and the data sample $X$ given by:

$$D(X,t) = \frac{2}{n}\sum_{i=1}^{n/2}\mathbf{1}\left\{ h(x_i,x_{i+n/2}) \leq t \right\}$$

\noindent This is exactly an empirical CDF for the distribution of $h(x,x^\prime)$, based on $n/2$ independent samples of $h(x,x^\prime)$.  We may anticipate that sample quantiles of this empirical CDF will be close to the quantiles from the true CDF, which we denote by $H(t) = P(h\leq t)$.  This is made precise by the DKW inequality (see e.g., \cite{massartdkw}), which in this case yields:

\begin{equation}
\label{dkw}
P\left( \sup_t |H(t) - D(X,t)| \geq \epsilon \right) \leq 2 e^{-n\epsilon^2}.
\end{equation}

Thus taking $d_\delta(X)$ to be the smallest $d$ with $D(X,d) = 1 - \delta$, and $h_{\delta^\prime}$ to give the $1-\delta^\prime$ quantile of $h$, with $\delta < \delta^\prime$, we have:

\begin{align*}
P(h(x,x^\prime) > d_\delta) &\leq \delta^\prime + P(d_\delta(X) < h_{\delta^\prime})\\
%&\leq \delta^\prime + P\left( D(X,h_{\delta^\prime}) - H(h_{\delta^\prime}) > (\delta^\prime-\delta)\right)\\
&\leq \delta^\prime + 2e^{-(\delta^\prime-\delta)^2n}.
\end{align*}

The second inequality comes from applying the monotone function $D(X,\cdot)$ to both sides of the inequality statement in the probability, and then rearranging, to yield $P\left( D(X,h_{\delta^\prime}) - H(h_{\delta^\prime}) > (\delta^\prime-\delta)\right)$ which is bounded due to the DKW inequality (\ref{dkw}).  Thus for some appropriate choice of $\delta,\delta^\prime$ we may take $s_n(X) = n^{-1}d_\delta(X)$, and thus achieve (\ref{rdp_cond2}).

Now to achieve (\ref{rdp_cond1}) we turn to the Bahadur-Kiefer representation of sample quantiles (see \cite{kieferquantile}).  We have that:

$$d_\delta(X) - h_{\delta} = \frac{D(X,h_\delta) - H(h_\delta)}{H^\prime(h_{\delta})} + O_p(n^{-3/4})$$

where $H^\prime$ is the derivative of $H$ (namely the density).  Hence we concentrate on the case when $h$ is a continuous random variable.  We find the ratio to be bounded in probability:

$$\frac{d_\delta(X)}{d_\delta(X^\prime)} \leq 1 + \frac{|d_\delta(X)-d_\delta(X^\prime)|}{d_\delta(X^\prime)} = 1 + \frac{O_p(n^{-1/2})}{h_\delta + O_p(n^{-1/2})}$$

where the final equality stems from using DKW to bound the $D(X,h_\delta) - H(h_\delta)$ and along with the triangle inequality to bound $|D(X,h_\delta) - D(X^\prime,h_\delta)|$.  This therefore demonstrates that:

$$\frac{d_\delta(X)}{d_\delta(X^\prime)} \leq 1 + O_p(n^{-1/2}) = O_p(e^{n^{-1/2}})$$

This means that for large enough $n$, and some probability $1-\gamma_2$, the ratio is bounded by $e^\beta$ where $\beta$ is polynomial in $n^{-1/2}$.  Examining (\ref{local_eta}) we find $\eta$ to be negligible for such a choice of $\beta$.  Therefore the use of $s_n(X) = n^{-1}d_\delta$ achieves the RDP as required.

We note that in principle this same approach would work, were we to replace $D(X,t)$ with the U-statistic process:

$$U(X,t) = \frac{1}{\binom{n}{2}}\sum_{i>j}\mathbf{1}\left\{ h(x_i,x_j) \leq t \right\}.$$

Though this is essentially another empirical CDF, it is based on non-independent samples since each $x_i$ participates in $n-1$ of the evaluations of $h$.  Nevertheless an analog of the DKW inequality still applies to this process, and we still have the same behavior of the quantiles (see e.g., \cite{
uquantile}).

\section{Privacy Concerns}

As stated above, we mainly use random differential privacy as a
vehicle for a theoretical exploration of the boundaries of differential
privacy.  Although it is a conceptually reasonable
weakening of differential privacy, whether it is appropriate to use in
practice requires more attention.  For example, if the hypothesized
adversary (of e.g., \cite{WZ} theorem 2.4), really had access to a
subset of $n-1$ of the data, and the one remaining element was the
only inhabitant of its histogram cell, then this would be immediately
revealed to the adversary.
%%%This is because the adversary's
%%%distribution over outputs for that cell is a point mass at zero,
%%%whereas the output from the agency with the full data would never
%%%contain zero in that position.
Whether this is a critical problem
depends on the application.

\section{Example}

We present two examples in which the RDP technique and the DP
techniques are compared on synthetic histogram data.  In the first
example the histogram has $k=25$ bins, all but two of which are empty
and $n=500$ points fall in to the other two.  Figure~\ref{fig_hist1}
shows the original data as well as the sanitized data due to
differential privacy and RDP.  Figure~\ref{fig_hist2} shows the
distribution of $L_1$ loss
from 100 simulations of both approaches.
We see that the risk of the RDP histogram is typically much
lower than that of the DP histogram, which occasionally has risk in
excess of 0.5 (recall that the maximum possible loss is 2 in the case
that the original and sanitized histograms had completely disjoint
support).

\begin{figure}[ht]
\centering
\subfigure[Original and synthetic data for DP (top) and RDP (bottom)]{
\includegraphics[scale=0.36]{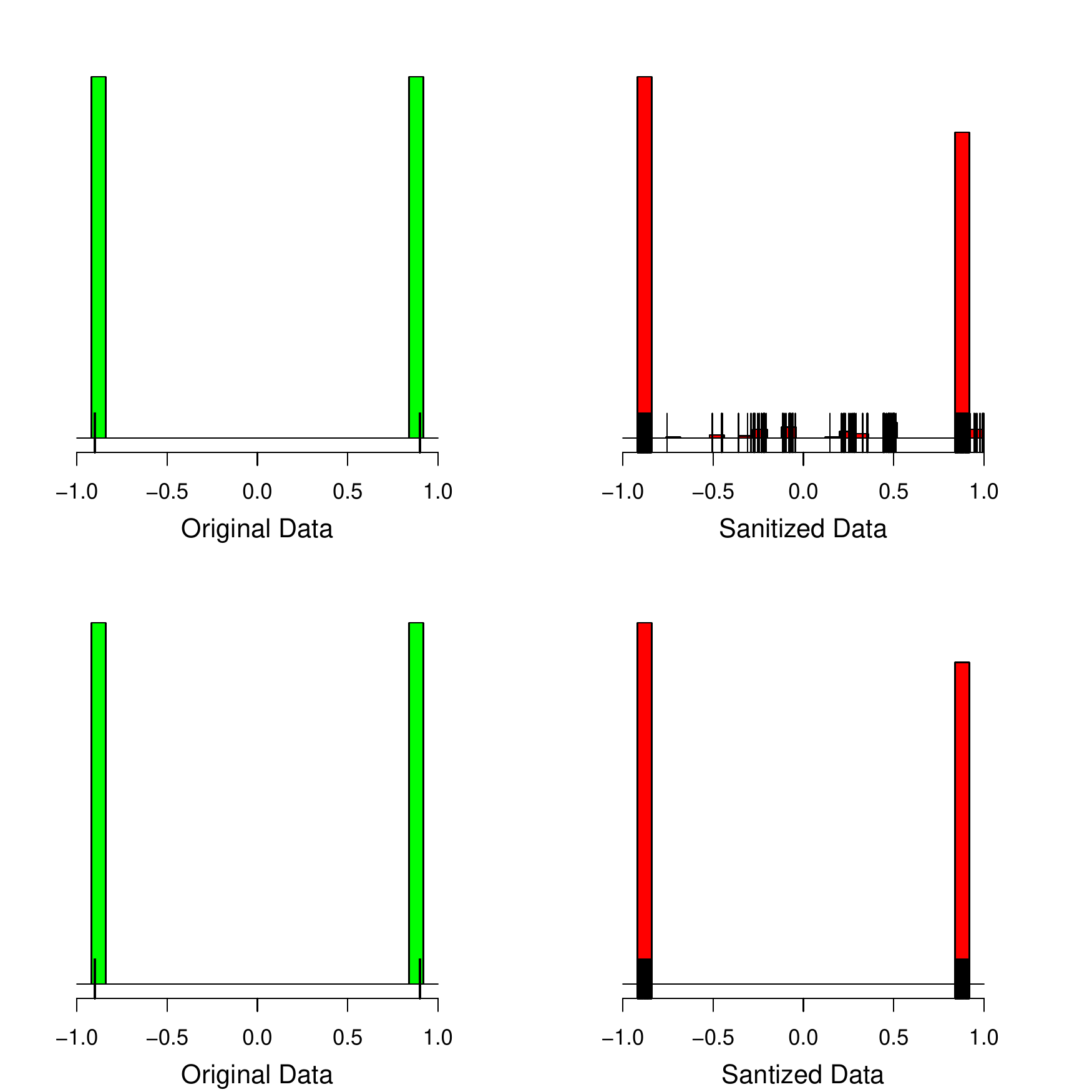}
\label{fig_hist1}
}
\subfigure[Empirical error distribution for DP (top) and RDP (bottom)]{
\includegraphics[scale=0.39]{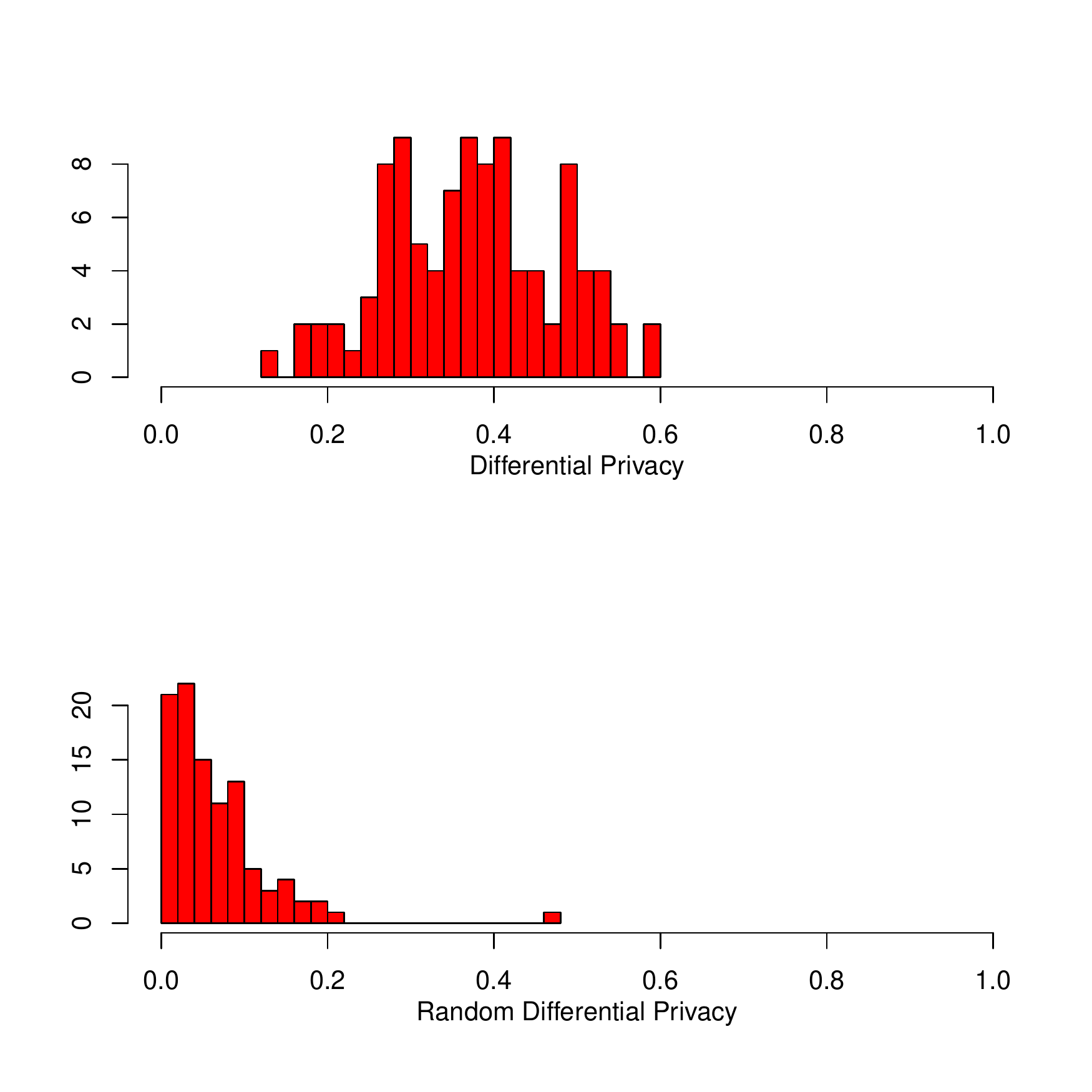}
\label{fig_hist2}
}
\label{fig_1dim}
\caption{A one dimensional example.}
\end{figure}

We present an analogous two dimensional example in
figure~\ref{fig_twodim}.  Here the histogram has
$k=400$ bins in which all but 16 are empty.
In this example we see that the RDP technique has
uniformly better loss than the DP technique.

\begin{figure}[ht]
\centering
\includegraphics[scale=0.5]{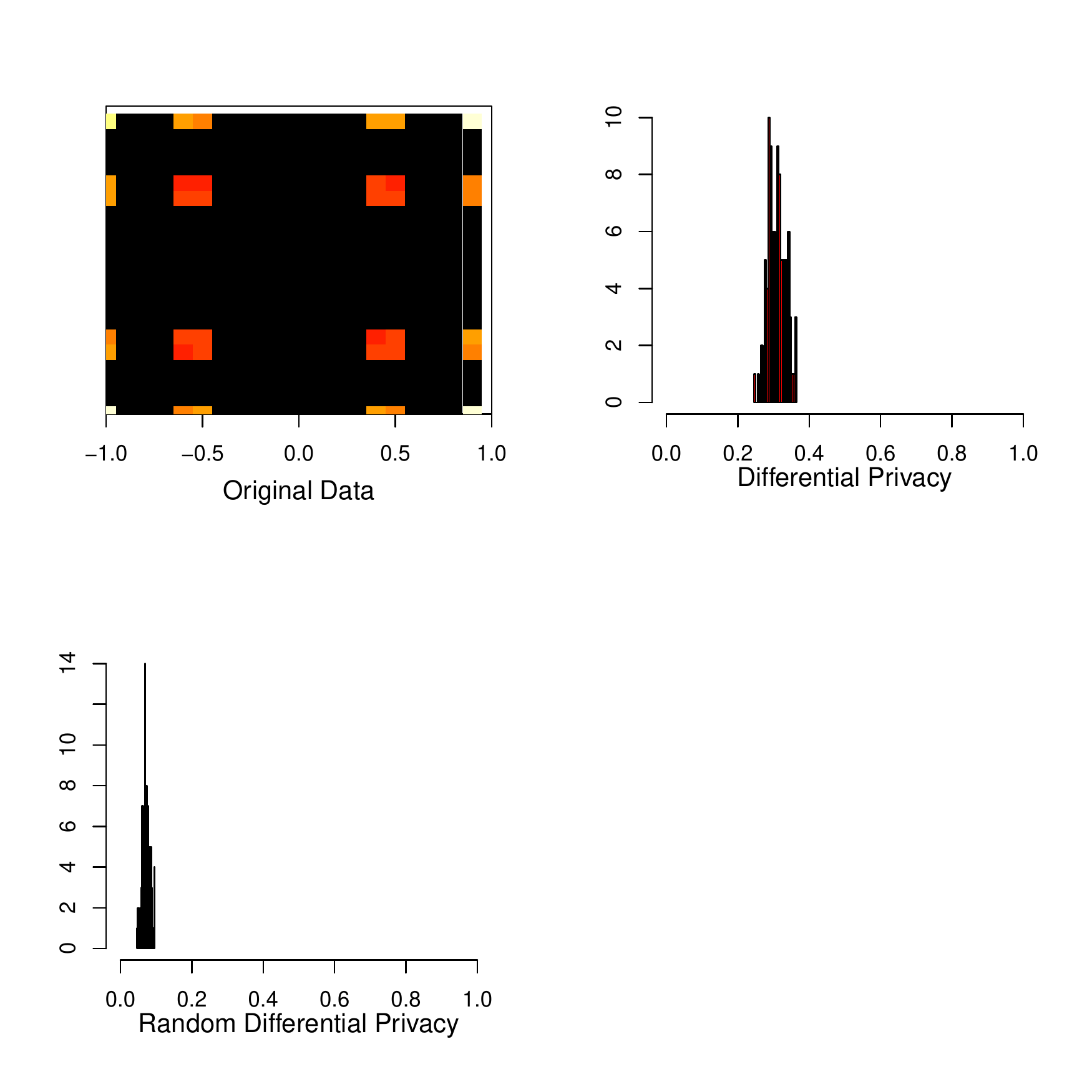}
\caption{Empirical error distributions
for a two dimensional histogram, displayed in the top left.}
\label{fig_twodim}
\end{figure}

\section{Conclusion}

We have introduced
a relaxed version of differential privacy---
random differential privacy---shown how to
apply it to histograms and examined the accuracy of the resulting method.  We also demonstrated some properties of our definition, and explained a basic construction for release of arbitrary functions of the data.
As we mentioned in the introduction,
we are not suggesting that
differential privacy should be abandoned and replaced by
random differential privacy.
However, we do think it is fruitful to consider various
relaxations of
differential privacy
to gain a deeper understanding
of the tradeoffs between
the strength of the privacy guarantee and the accuracy
of the data release mechanism.

In ongoing work we are extending this work to
allow for data dependent choices of the number of bins and
to allow for other density estimators besides histograms.
We are also considering other relaxations of
differential privacy.
We will report on these results in future work.

\subsubsection*{References}
\vspace{-11mm}
\bibliographystyle{plain}
\renewcommand{\refname}{}
\bibliography{sdp}

\end{document}